\documentclass[12pt]{article}
\usepackage[letterpaper,margin=1in]{geometry}
\usepackage{amsmath,amssymb,amsthm,mathtools}
\usepackage{microtype}
\usepackage{titlesec}
\usepackage[hidelinks]{hyperref}
\titleformat{\section}{\bfseries\large}{\thesection.}{0.45em}{}
\titleformat{\subsection}{\bfseries\normalsize}{\thesubsection.}{0.45em}{}
\titlespacing*{\section}{0pt}{8pt}{3pt}
\titlespacing*{\subsection}{0pt}{6pt}{2pt}

\theoremstyle{plain}
\newtheorem{theorem}{Theorem}
\newtheorem{corollary}{Corollary}
\newtheorem{lemma}{Lemma}
\newtheorem{assumption}{Assumption}
\theoremstyle{remark}
\newtheorem{remark}{Remark}

\newcommand{\uL}{\underline\Lambda}
\newcommand{\un}{\underline\nu}
\newcommand{\PP}{\mathcal P}

\newcommand{\TT}{\mathcal T}
\newcommand{\Hs}{H}

\title{\bfseries Stabilization of a Heterodirectional\\ Bradytachic (1D, 2D) PDE Pair}
\author{Miroslav Krstic\thanks{Department of Mechanical and Aerospace Engineering, University of California, San Diego; mkrstic@ucsd.edu.}}
\date{}

\begin{document}

\maketitle

\begin{abstract}
A scalar hyperbolic PDE, actuated at one boundary, is coupled with a fast PDE in two spatial dimensions: transport in the axial coordinate and in an internal coordinate, and possibly diffusion in the internal coordinate. Backstepping designs exist for coupled hyperbolic systems, for their ensembles, and for their continua, but in all of these the second variable of the fast subsystem carries no transport or diffusion; when it does, no design is available, and from a scalar input no exact control of the two-dimensional state is to be expected. The pair is stabilized here by time-scale separation: in the quasi-steady limit the two-dimensional subsystem collapses into a spatial Volterra operator inside a one-dimensional reduced plant, backstepping applies there, and the two-dimensional subsystem is left to be a stable boundary layer. One theorem establishes exponential stability of the pair for every time-scale ratio below an explicit threshold, uniformly in the internal diffusion coefficient down to the pure-transport case.
\end{abstract}

\section{Introduction}

\subsection{Open problem: from ensembles to 2D fast PDEs}
Backstepping stabilization of a pair of counterconvecting transport PDEs was introduced in \cite{VazquezKrsticCoron2011,CoronVazquezKrsticBastin2013}, extended to $n+1$ systems actuated by a single boundary input \cite{DiMeglioVazquezKrstic2013} and to general $n+m$ systems \cite{HuDiMeglio2016}, and refined to achieve the minimum stabilization time of \cite{LiRao2010} in \cite{AuriolDiMeglio2016,CoronHuOlive2017}; delay-robustness is analyzed in \cite{AuriolAarsnes2018}; states sharing a transport speed (isotachic) are handled by the decoupling of \cite{HuDiMeglio2016}; the book \cite{BastinCoron2016} and the survey \cite{VazquezAuriol2026} cover the hyperbolic theory and PDE backstepping at large.

Recent work treats systems whose number of states is large or infinite. Alleaume and Krstic \cite{AlleaumeKrstic2025} stabilize a scalar hyperbolic PDE, actuated at one boundary, coupled with an ensemble of counterconvecting transport PDEs: a continuum of one-dimensional equations indexed by a parameter $z$, with bounded couplings across the ensemble. Humaloja and Bekiaris-Liberis develop backstepping for continua of hyperbolic PDEs as a tool for the large-scale finite case: the control law designed for the continuum stabilizes the large but finite $n+m$ system \cite{HumalojaBLTAC2025}, the design extends to continua of rightward transporting PDEs coupled with finitely many counterconvecting states, with kernel equations on a prismatic three-dimensional domain \cite{HumalojaBLAut2025}, the continuum kernels are computable with complexity independent of the number of components \cite{HumalojaBLSCL2025}, and output-feedback versions rest on continuum observers \cite{HumalojaBLAut2026of}. In all of these, there is no transport or diffusion in the ensemble variable: $z$ enters the kernels as a parameter, never as a direction of propagation. That is what keeps these systems within reach of a backstepping design with one, or finitely many, boundary inputs.

The present paper poses and solves the next problem. The plant is the pair
\begin{equation}
 u_t(x,t)
 =u_x(x,t)+\int_0^D c(x,z)v(x,z,t)\,dz,
 \qquad 0<x<1,
 \label{eq:gen-u}
\end{equation}
\begin{align}
 &\delta v_t(x,z,t)+\Lambda(x)v_x(x,z,t)+\nu(x,z)v_z(x,z,t)-\varepsilon v_{zz}(x,z,t)
 \nonumber\\
 &\qquad\qquad=d(x,z)v(x,z,t)+b(x,z)u(x,t)
 \nonumber\\
 &\qquad\qquad\hphantom{=}
 +\int_0^D\theta(x,z,\zeta)v(x,\zeta,t)\,d\zeta,
 \qquad 0<z<D,
 \label{eq:gen-v}
\end{align}
with only the scalar boundary input $u(1,t)=U(t)$ available. The terms $\nu(x,z)v_z$ and, when present, $\varepsilon v_{zz}$ are absent from the ensemble: they make $v$ a single two-dimensional PDE on the rectangle $(0,1)\times(0,D)$, transporting (and possibly diffusing) through the internal coordinate, rather than a continuum of one-dimensional PDEs that share an index. No backstepping design is available for a two-dimensional PDE actuated through the boundary trace of its one-dimensional neighbor, and from a scalar input no form of exact controllability of $v$ is to be expected.

\subsection{Solution idea}
The solution has two ingredients.

The first is time-scale separation: the parameter $\delta$ in \eqref{eq:gen-v} is small, so the two-dimensional subsystem is fast. Then $v$ does not need to be controlled at all. It needs only to be a stable boundary layer around its quasi-steady state, and the entire design can take place on the reduced model obtained at $\delta=0$.

The second is the conversion of the quasi-steady two-dimensional PDE into a Volterra operator inside a single one-dimensional PDE. At $\delta=0$ the fast equation loses its time derivative and becomes, for each frozen $t$, an evolution in the space variable $x$: with zero inflow at $x=0$ and at $z=0$, it is marched from the zero state, driven by the slow profile $u(\cdot,t)$. Eliminating the fast state converts the pair into the scalar plant
\begin{equation}
 u_t(x,t)=u_x(x,t)+\int_0^x f(x,\xi)\,u(\xi,t)\,d\xi ,
 \label{eq:introreduced}
\end{equation}
a hyperbolic PDE with a causal spatial Volterra operator, which backstepping stabilizes with the one available input \cite{KrsticSmyshlyaev2008}. The paper then proceeds as follows: backstepping on the reduced Volterra plant; the resulting boundary feedback applied, unchanged, to the true $\delta>0$ plant; and a composite Lyapunov functional, with closed-form optimized slow-fast weighting, certifying exponential stability of the full coupled system for every $\delta$ below an explicit threshold, uniformly in the internal diffusion coefficient $\varepsilon\in[0,\bar\varepsilon]$ down to the pure-transport case $\varepsilon=0$ (Theorem~\ref{thm:main}, with all constants explicit in the plant data).

The conversion depends on causality. A quasi-steady PDE generically produces a spatial boundary-value problem, whose solution at $x$ depends on the input on both sides of $x$: a Fredholm, not a Volterra, operator. That is the situation in the linear parabolic-elliptic PDE differential-algebraic couplings of Alalabi--Morris \cite{AlalabiMorris2024,AlalabiMorris2025}, and no backstepping design of the present kind is available there. What makes the operator here Volterra is the marching structure: the fast subsystem transports away from $x=0$ with zero inflow, and diffusion, when present, acts transversally to the marching direction. Two earlier constructions contain the same conversion one dimension down, with quasi-steady spatial ODEs in place of a quasi-steady spatial PDE: in the Korteweg--de Vries-like system of \cite{KrsticSmyshlyaevBook2008,KrsticSmyshlyaev2008}, the second boundary control zeroes the fast state and its derivative at $x=0$, turning the quasi-steady second-order ODE into a zero-state initial-value problem and the eliminated variable into a strict spatial Volterra term; in the thermal convection loop of \cite{VazquezKrstic2006}, boundary actuation likewise removes the noncausal component and the quasi-steady fast velocity becomes an explicit causal integral of the temperature. The step taken here is from spatial ODEs to a two-dimensional spatial PDE, with the boundary-layer stability analysis to match.

\paragraph{Contributions.}
\begin{itemize}
\item Heterodirectional PDE pairs whose fast subsystem is two-dimensional---transporting and possibly diffusing in a second spatial variable---have been outside the reach of boundary control design. This paper stabilizes them with the single scalar input available.
\item In the only previous combination of singular perturbation with PDE backstepping, the thermal convection loop \cite{VazquezKrstic2006}, the quasi-steady state solves an ordinary differential equation, in closed form. Here it is the solution operator of a two-dimensional PDE, and its substitution produces a one-dimensional plant of the Volterra class of \cite{KrsticSmyshlyaev2008}, with kernels that no ordinary-differential or ensemble fast subsystem generates.
\item The obstacle that a PDE-produced quasi-steady state raises, and that an ODE-produced one does not, is its motion in time, which injects the spatial derivative of the slow state into the fast subsystem, a term for which an $L^2$ analysis has no norm. The central estimate of the proof shows that this derivative, integrated against the quasi-steady kernel, amounts to no more than the slow state and its outflow trace---and the trace is absorbed by the boundary dissipation of the transport target, the term that Lyapunov analyses of hyperbolic systems ordinarily discard.
\item The design is, to our knowledge, the first stabilizing feedback for a partial differential-algebraic system (PDAE) whose algebraic part is itself a PDE, and certainly the first by PDE backstepping.
\end{itemize}
A single feedback law and a single theorem cover the fast subsystems from pure transport to ultraparabolic transport-diffusion.

\paragraph{Terminology.}
The title terminology extends the naming convention of \cite{HuDiMeglio2016}, where \emph{heterodirectional} was introduced for coupled hyperbolic systems whose transport speeds have opposite signs, and \emph{isotachic} for states whose transport speeds coincide. The present pair is heterodirectional---$u$ convects toward decreasing $x$, $v$ toward increasing $x$---and \emph{bradytachic}, from the Greek $\beta\rho\alpha\delta\acute\upsilon\varsigma$ (slow) and $\tau\alpha\chi\acute\upsilon\varsigma$ (fast): one subsystem of the pair is slow and the other fast. As in the clinical usage that makes the two prefixes familiar, they denote extremes rather than mild deviations, so the word carries not only that the transport speeds differ but that they are separated by a wide margin, quantified by the ratio $\delta$ of the two time scales. The separation is what makes the problem solvable. A heterodirectional (1D, 1D) pair needs no time-scale separation at all, since coupled-hyperbolic backstepping stabilizes it exactly for any speed ratio. For the fast subsystem in two dimensions, the speed separation reduces the demand from control to boundary-layer stability, and Theorem~\ref{thm:main} quantifies a width of separation that suffices.

\paragraph{Notation.}
$\Hs:=L^2(0,D)$ with norm $\|\cdot\|_\Hs$ and inner product $\langle\cdot,\cdot\rangle_\Hs$. For $u\in L^2(0,1)$, $\|u\|$ denotes the $L^2(0,1)$ norm. For $\Hs$-valued functions of $x$ we write $\|F\|_{L^2_x(\Hs)}^2=\int_0^1\|F(x)\|_\Hs^2\,dx$ and identify $L^2_x(\Hs)$ with $L^2((0,1)\times(0,D))$; the same symbol $\|\cdot\|$ is used when no confusion arises. Time arguments are suppressed where the operations act at frozen $t$.

\section{Problem Formulation}
The plant \eqref{eq:gen-u}, \eqref{eq:gen-v} is posed on $0<x<1$, $0<z<D$, with a small parameter $0<\delta\ll1$, an internal diffusion coefficient $\varepsilon\ge0$, boundary conditions
\begin{equation}
 u(1,t)=U(t),\qquad v(0,z,t)=0,
 \label{eq:gen-bc}
\end{equation}
and, in the internal coordinate,
\begin{equation}
 \nu(x,0)\,v(x,0,t)-\varepsilon v_z(x,0,t)=0,
 \qquad
 v_z(x,D,t)=0\quad(\varepsilon>0),
 \label{eq:gen-zbc}
\end{equation}
with initial conditions $u(x,0)=u_0(x)$, $v(x,z,0)=v_0(x,z)$. The flux in $z$ is $\nu v-\varepsilon v_z$, so the first condition in \eqref{eq:gen-zbc} is zero inflow flux at $z=0$; as $\varepsilon\to0$ it degenerates to the Dirichlet inflow condition $v(x,0,t)=0$, and the Neumann condition at $z=D$ is void at $\varepsilon=0$, where the operator is first order in $z$. The slow field $u$ travels toward decreasing $x$; the fast field $v$ travels toward increasing $x$ and increasing $z$. The unit coefficient in \eqref{eq:gen-u} is without loss of generality: any positive scalar speed $\lambda(x)$ is removed by a monotone reparametrization of $x$ and a constant rescaling of time. Taking the fast axial speed as $\Lambda(x)$ rather than an arbitrary $\Lambda(x,z)$ is a structural assumption, not a normalization; the speed $\nu(x,z)$ in the internal coordinate is left fully two-dimensional.

Physical readings of $z$. At $\varepsilon=0$, $z$ is an age or maturity coordinate of the fast species, injected from $u$ with intensity $b$, aggregated through $c$, growing or decaying through $d$, and coupled across ages through $\theta$. For $\varepsilon>0$ the fast equation is hyperbolic in $x$ and parabolic in $z$, an ultraparabolic equation with the two time-like directions $t$ and $x$, and two further readings open: $z$ a transverse channel coordinate, giving the convection-diffusion balance of a plug-flow reactor or counterflow heat exchanger with transverse diffusion, the Graetz configuration; and $z$ a physiological trait, with $-\nu v_z+\varepsilon v_{zz}$ the drift-diffusion part of the Fokker--Planck operator $-\partial_z(\nu v)+\varepsilon v_{zz}$ of deterministic maturation plus random trait drift, whose zero-order term $-\nu_zv$ is absorbed into $d(x,z)v$. The design depends on causality in $x$ surviving: diffusion transverse to the marching direction leaves the quasi-steady problem a well-posed evolution in $x$ from the zero state, so the reduced model is a causal spatial Volterra operator; diffusion along $x$ would instead produce a two-sided boundary-value problem and a Fredholm operator, the situation of \cite{AlalabiMorris2024,AlalabiMorris2025}.

\begin{assumption}\label{as:coeff}
$\Lambda\in C^1([0,1])$ and $\nu\in C^1([0,1]\times[0,D])$, with $\Lambda(x)\ge\uL>0$ and $\nu(x,z)\ge\un>0$; $b\in C^1([0,1]\times[0,D])$, with $b(x,\cdot)\in C^2([0,D])$ uniformly in $x$; and $c$, $d$, $\theta$ are continuous.
\end{assumption}

The substance of Assumption~\ref{as:coeff} is the regularity of the coefficients and the uniform positivity of the two transport speeds; everything else is notation. With $\rho:=\nu/\Lambda$, the remaining symbols are shorthands for norms of the data, finite by Assumption~\ref{as:coeff} on the compact rectangle:
\begin{gather}
 \Lambda_1:=\|\Lambda'\|_\infty,\qquad
 \nu_1:=\|\nu_z\|_\infty,\qquad
 \bar\Lambda:=\|\Lambda\|_\infty,\qquad
 \bar\nu:=\|\nu\|_\infty,
 \nonumber\\
 \bar\rho:=\|\rho\|_\infty,\qquad
 \rho_1:=\|\rho_z\|_\infty\le\frac{\nu_1}{\uL} ,
 \label{eq:speedconst}\\
 B_0:=\|b\|_\infty,\qquad
 B_x:=\|b_x\|_\infty,\qquad
 B_z:=\|b_z\|_\infty,\qquad
 B_{zz}:=\|b_{zz}\|_\infty,
 \label{eq:bconst}\\
 M_C:=\sup_{x}\|c(x,\cdot)\|_{L^2(0,D)},
 \qquad
 \bar d:=\|d\|_\infty,
 \qquad
 \bar\theta:=\sup_{x}\|\theta(x,\cdot,\cdot)\|_{L^2((0,D)^2)} .
 \label{eq:cdthetaconst}
\end{gather}
The diffusion coefficient satisfies $0\le\varepsilon\le\bar\varepsilon$ for a fixed $\bar\varepsilon\ge0$; all constants below depend on $\bar\varepsilon$ but not on $\varepsilon$.

The closed-loop system, once the feedback \eqref{eq:genU} below is inserted, is a linear system of transport ($\varepsilon=0$) or transport-diffusion ($\varepsilon>0$) equations with bounded interior couplings and a boundary condition given by a bounded functional of the state. Its well-posedness---a unique mild solution for every $\delta>0$---is part of Theorem~\ref{thm:main}. For smooth compatible initial data the solutions are piecewise classical, with weak discontinuities only along characteristics at $\varepsilon=0$; all Lyapunov computations below are performed on such solutions and extend by density. The control objective is exponential stability of the origin in the $L^2$ topology, for all sufficiently small $\delta$, with an explicit admissible range of $\delta$, uniform in $\varepsilon\in[0,\bar\varepsilon]$.

\section{Reduced Model}\label{sec:design}

Set $\delta=0$ in \eqref{eq:gen-v}, divide by $\Lambda(x)$, and regard $V(x):=v(x,\cdot)\in\Hs$ as evolving in $x$:
\begin{equation}
 V_x(x)=\mathcal A_\varepsilon(x)V(x)+B(x)u(x),\qquad V(0)=0,
 \label{eq:Hspatial}
\end{equation}
where
\begin{align}
 \mathcal A_\varepsilon(x)\phi
 &:=\frac{1}{\Lambda(x)}
 \Big[\varepsilon\phi''-\nu(x,\cdot)\phi'+d(x,\cdot)\phi
 +\int_0^D\theta(x,\cdot,\zeta)\phi(\zeta)\,d\zeta\Big],
 \label{eq:Aeps}\\
 \big(B(x)r\big)(z)&:=\frac{b(x,z)}{\Lambda(x)}\,r,
 \qquad
 C(x)\phi:=\int_0^D c(x,z)\phi(z)\,dz,
 \label{eq:BCgen}
\end{align}
on the domain $\{\phi\in H^2(0,D):\varepsilon\phi'(0)=\nu(x,0)\phi(0),\ \phi'(D)=0\}$ for $\varepsilon>0$, and the transport domain $\{\phi\in H^1(0,D):\phi(0)=0\}$ for $\varepsilon=0$. Split $\mathcal A_\varepsilon(x)=\mathcal A_{\varepsilon,0}(x)+\mathcal R(x)$, where $\mathcal A_{\varepsilon,0}(x)$ collects the derivative terms and
\begin{equation}
 \mathcal R(x)\phi
 :=\frac{d(x,\cdot)}{\Lambda(x)}\phi
 +\frac{1}{\Lambda(x)}\int_0^D\theta(x,\cdot,\zeta)\phi(\zeta)\,d\zeta,
 \qquad
 \|\mathcal R(x)\|_{\mathcal L(\Hs)}\le M_R:=\frac{\bar d+\bar\theta}{\uL}.
 \label{eq:Rdef}
\end{equation}

For $0\le\xi\le x\le1$ and $\phi\in\Hs$, define $\Phi_\varepsilon(x,\xi)\phi:=V(x)$, where $V$ solves the homogeneous part of \eqref{eq:Hspatial}, namely $V_x=\mathcal A_\varepsilon(x)V$ on $(\xi,1)$ with $V(\xi)=\phi$: a zero-inflow transport problem, solved along characteristics, for $\varepsilon=0$, and a parabolic evolution for $\varepsilon>0$. The zero-state solution of \eqref{eq:Hspatial} itself is then
\begin{equation}
 (\PP_\varepsilon u)(x)=\int_0^x\Phi_\varepsilon(x,\xi)B(\xi)u(\xi)\,d\xi ,
 \label{eq:Pgen}
\end{equation}
which satisfies the homogeneous conditions \eqref{eq:gen-zbc} and vanishes at $x=0$. We call $\PP_\varepsilon$ the quasi-steady-state linear operator: a linear operator mapping functions of the single spatial variable $x$ into functions of the two spatial variables $x$ and $z$. Uniformly in $\varepsilon$,
\begin{equation}
 \|\Phi_\varepsilon(x,\xi)\|_{\mathcal L(\Hs)}\le M_\Phi:=e^{\rho_1/2+M_R},
 \qquad
 \|\PP_\varepsilon\|_{\mathcal L(L^2(0,1),L^2_x(\Hs))}\le M_P:=\frac{M_\Phi\sqrt D\,B_0}{\uL} .
 \label{eq:Phibound}
\end{equation}
Well-posedness of both solution maps and the bounds \eqref{eq:Phibound} are routine; they are certified by Lemma~\ref{lem:evolution} in the appendix.

Elimination of the fast state at $\delta=0$ thus produces the reduced slow plant
\begin{equation}
 u_t(x,t)=u_x(x,t)+\int_0^xf_\varepsilon(x,\xi)u(\xi,t)\,d\xi,
 \qquad u(1,t)=U(t),
 \label{eq:genreslow}
\end{equation}
with the nonstationary spatial Volterra kernel
\begin{equation}
 f_\varepsilon(x,\xi):=C(x)\Phi_\varepsilon(x,\xi)B(\xi),
 \qquad
 |f_\varepsilon(x,\xi)|\le\bar f:=M_CM_P,
 \label{eq:Kgen}
\end{equation}
continuous on the triangle $T:=\{0\le\xi\le x\le1\}$---despite the jump in $z$ that the $\Hs$-valued kernel $\Phi_\varepsilon(x,\xi)B(\xi)$ carries at $\varepsilon=0$ across a moving characteristic; the continuity is proved with Lemma~\ref{lem:evolution}.

\section{Main Result}\label{sec:main}

\begin{theorem}\label{thm:main}
Let Assumption~\ref{as:coeff} hold, let $\varepsilon\in[0,\bar\varepsilon]$, and apply to the plant \eqref{eq:gen-u}--\eqref{eq:gen-zbc} the boundary feedback
\begin{equation}
 U(t)=\int_0^1k_\varepsilon(1,\xi)u(\xi,t)\,d\xi ,
 \label{eq:genU}
\end{equation}
where $k_\varepsilon$ is the unique continuous solution on $T$ of
\begin{equation}
 \partial_xk_\varepsilon(x,\xi)+\partial_\xi k_\varepsilon(x,\xi)
 =\int_\xi^xk_\varepsilon(x,s)f_\varepsilon(s,\xi)\,ds-f_\varepsilon(x,\xi),
 \qquad k_\varepsilon(x,0)=0.
 \label{eq:genq}
\end{equation}
Then, for every $\delta>0$, the closed loop has a unique mild solution, continuous in time into $L^2(0,1)\times L^2((0,1)\times(0,D))$, and, for every $\delta$ obeying
\begin{equation}
 0<\delta<\Delta:=\sup_{\lambda>0,\ \sigma>\sigma_0}\delta^\star(\lambda,\sigma),
 \label{eq:deltarange}
\end{equation}
where
\begin{equation}
 \delta^\star(\lambda,\sigma)
 :=
 \frac{a_f}
 {\displaystyle
 \gamma+\frac{\beta_1}{\lambda}
 \left(\beta_2+\sqrt{\beta_2^2+\frac{\lambda\beta_0^2}{2}}\right)} ,
 \label{eq:deltastar}
\end{equation}
there exist $M\ge1$ and $\alpha>0$, independent of $\varepsilon\in[0,\bar\varepsilon]$, such that the solution obeys
\begin{equation}
 \|u(\cdot,t)\|^2+\|v(\cdot,\cdot,t)\|^2
 \le
 M e^{-\alpha t}
 \left(\|u_0\|^2+\|v_0\|^2\right),
 \qquad t\ge0 ,
 \label{eq:mainestimate}
\end{equation}
in the norms of $L^2(0,1)$ and $L^2((0,1)\times(0,D))$, respectively, where, for $\lambda>0$ and $\sigma>\sigma_0$,
\begin{gather}
 \sigma_0:=\frac{\Lambda_1+\nu_1+2\bar d+2\bar\theta}{\uL},
 \qquad
 a_s:=\lambda,
 \qquad
 a_f:=\tfrac12\big(\sigma\uL-\Lambda_1-\nu_1-2\bar d-2\bar\theta\big),
 \label{eq:asaf}\\
 \gamma:=C_3e^{\sigma/2},
 \qquad
 \beta_1:=(1+\bar k)M_C\,e^{(\lambda+\sigma)/2},
 \qquad
 \beta_2:=C_1,
 \qquad
 \beta_0:=\sqrt2\,C_2 ,
 \label{eq:betas}\\
 C_1:=M_T\big(A_P+M_P^2M_C\big),
 \qquad
 C_2:=B_P,
 \qquad
 C_3:=M_PM_C ,
 \label{eq:C123}\\
 \bar k:=e^{\bar f}-1,
 \qquad
 M_T:=1+\bar k\,e^{\bar k},
 \label{eq:kbarMT}\\
 B_P:=\frac{M_\Phi\sqrt D\,B_0}{\uL},
 \label{eq:BP}\\
 A_P:=M_\Phi\sqrt D\left[
 \frac{B_x+B_0\Lambda_1/\uL}{\uL}
 +\frac{\bar\varepsilon B_{zz}+\bar\nu B_z+\bar dB_0+\bar\theta B_0}{\uL^2}
 \right]
 \nonumber\\
 +\frac{\sqrt D\,B_0}{\uL}
 +N_\Psi\,\frac{B_0\bar\nu+2\bar\varepsilon B_z}{\uL^2},
 \label{eq:AP}\\
 N_\Psi:=\left(\frac{2\bar\Lambda\big[(\tfrac{\rho_1}{2}+M_R)M_\Phi^2+M_\Phi\big]}{\un}\right)^{1/2}.
 \label{eq:NPsi}
\end{gather}
\end{theorem}

The weights $\lambda$ and $\sigma$ are analysis parameters only; the controller \eqref{eq:genU} does not depend on them, which is what legitimizes the supremum in \eqref{eq:deltarange}. For fixed geometry, $\Delta$ depends only on the coefficients $b,c,d,\theta,\Lambda,\nu$, their regularity bounds, and $\bar\varepsilon$. The kernel $k_\varepsilon$ is designed at the known value of $\varepsilon$; what is uniform in $\varepsilon$ is the admissible range \eqref{eq:deltarange} and the constants $M,\alpha$, not the controller.

To a reader mindful of how differently parabolic and hyperbolic dynamics behave, a single theorem spanning $\varepsilon=0$ and $\varepsilon>0$ may look suspicious. It is not: the claim is an $L^2$ energy statement, blind to the regularity distinction between the two regimes---under the flux and Neumann conditions \eqref{eq:gen-zbc} the diffusion enters every estimate only dissipatively---and the distinction surfaces only inside the proof, in the kernel of $\PP_\varepsilon$, which has a jump along a characteristic at $\varepsilon=0$ and is smoothing for $\varepsilon>0$; the adjoint argument of Lemma~\ref{lem:Pux} treats both alike. The hyperbolic case is still worth recording on its own, and it comes with weaker hypotheses and smaller constants.

\begin{corollary}\label{cor:transport}
Let $\bar\varepsilon=0$, so that the fast subsystem \eqref{eq:gen-v} is a pure transport PDE with the inflow conditions $v(0,z,t)=0$, $v(x,0,t)=0$, and let Assumption~\ref{as:coeff} hold with $b\in C^1$ only. Then Theorem~\ref{thm:main} holds verbatim with
\begin{equation}
 A_P=M_\Phi\sqrt D\left[
 \frac{B_x+B_0\Lambda_1/\uL}{\uL}
 +\frac{\bar\nu B_z+\bar dB_0+\bar\theta B_0}{\uL^2}
 \right]
 +\frac{\sqrt D\,B_0}{\uL}
 +N_\Psi\,\frac{B_0\bar\nu}{\uL^2} .
 \label{eq:APtransport}
\end{equation}
\end{corollary}

\begin{proof}
Set $\varepsilon=0$ throughout the proof of Theorem~\ref{thm:main}: every term carrying $\varepsilon$ vanishes, and $b_{zz}$ never arises, so $b(x,\cdot)\in C^1$ suffices.
\end{proof}

\begin{remark}
No dissipativity or smallness assumption on $d$ or $\theta$ is made: only boundedness. The fast generator is multiplied by the $1/\delta$ time scale, so a sufficiently large weight decay rate $\sigma$ in the axial direction makes the transport outflow dominate the bounded zero-order and nonlocal terms, as \eqref{eq:asaf} shows; the weight is taken in $x$ only because a weight decaying in $z$ would be incompatible with a nonlocal coupling $\theta$ acting across the whole interval $(0,D)$.
\end{remark}

\section{Proof of Theorem \ref{thm:main}}\label{sec:proof}
Existence and uniqueness of the closed-loop solution are Lemma~\ref{lem:wellposed} in the appendix; this section proves the estimate \eqref{eq:mainestimate}.

\subsection{Transformation, its identity, and exact target system}
Introduce the Volterra transformation
\begin{equation}
 w(x,t)=(\TT u)(x,t):=u(x,t)-\int_0^xk_\varepsilon(x,\xi)u(\xi,t)\,d\xi ;
 \label{eq:genT}
\end{equation}
$\TT$ is boundedly invertible on $L^2(0,1)$, uniformly in $\varepsilon$ (Lemma~\ref{lem:invert} in the appendix). If $u$ satisfies $u_t=u_x+\int_0^xf_\varepsilon(x,\xi)u(\xi)\,d\xi+g(x)$ with any forcing $g\in L^2(0,1)$, then, integrating $\int_0^xk_\varepsilon(x,\xi)u_\xi(\xi)\,d\xi$ by parts and inserting \eqref{eq:genq},
\begin{equation}
 w_t=w_x+(\TT g)(x),
 \label{eq:Tidentity}
\end{equation}
and \eqref{eq:genU} enforces $w(1,t)=0$; for kernels that are only continuous the computation is justified in the weak formulation, exactly as in the $L^\infty$-kernel setting of \cite{HuDiMeglio2016}. The lower limit in \eqref{eq:genT} gives the trace identity
\begin{equation}
 w(0,t)=u(0,t).
 \label{eq:traceidentity}
\end{equation}
At $\delta=0$, where $g=0$, this establishes the design claim of Section~\ref{sec:main}: the reduced closed loop is the exponentially stable transport equation $w_t=w_x$, $w(1,t)=0$.

For $\delta>0$, define the quasi-steady-state error
\begin{equation}
 \eta(x,z,t):=v(x,z,t)-(\PP_\varepsilon u)(x,z,t),
 \label{eq:etadef}
\end{equation}
where $\PP_\varepsilon$ acts on $u(\cdot,t)$ at each frozen $t$. By Lemma~\ref{lem:evolution},
\begin{equation}
 \eta(0,z,t)=0,
 \qquad
 \nu(x,0)\eta(x,0,t)-\varepsilon\eta_z(x,0,t)=0,
 \qquad
 \eta_z(x,D,t)=0\quad(\varepsilon>0).
 \label{eq:etabc}
\end{equation}
Since $v=\PP_\varepsilon u+\eta$ and $C(x)(\PP_\varepsilon u)(x)=\int_0^xf_\varepsilon(x,\xi)u(\xi)\,d\xi$, the slow equation \eqref{eq:gen-u} reads
\begin{equation}
 u_t=u_x+\int_0^xf_\varepsilon(x,\xi)u(\xi,t)\,d\xi+C(\cdot)\eta,
 \label{eq:usplit}
\end{equation}
and the identity \eqref{eq:Tidentity} with $g=C(\cdot)\eta$ gives the exact slow target
\begin{equation}
 w_t=w_x+\TT\big[C(\cdot)\eta\big],\qquad w(1,t)=0 .
 \label{eq:wtarget}
\end{equation}
The quasi-steady state satisfies, pointwise a.e.\ in the rectangle,
\begin{equation}
 \Lambda(\PP_\varepsilon u)_x+\nu(\PP_\varepsilon u)_z-\varepsilon(\PP_\varepsilon u)_{zz}
 =d\,(\PP_\varepsilon u)+\int_0^D\theta(x,\cdot,\zeta)(\PP_\varepsilon u)(x,\zeta)\,d\zeta+b\,u,
 \label{eq:qss}
\end{equation}
which is \eqref{eq:Hspatial} written in the original variables. Subtracting \eqref{eq:qss} from \eqref{eq:gen-v} and using $\partial_t\PP_\varepsilon u=\PP_\varepsilon u_t$ (the kernel of $\PP_\varepsilon$ is time-independent) leaves the quasi-steady-state motion term $-\delta\,\PP_\varepsilon u_t$ in the fast equation. This term is not left in that form: substituting $u_t$ from the closed-loop slow equation \eqref{eq:usplit} and $u=\TT^{-1}w$ expresses it through the transformed states alone, and the exact fast target becomes
\begin{equation}
 \delta\eta_t+\Lambda\eta_x+\nu\eta_z-\varepsilon\eta_{zz}
 =d\,\eta+\int_0^D\theta(x,\cdot,\zeta)\eta(x,\zeta,t)\,d\zeta
 -\delta\,\PP_\varepsilon C(\cdot)\eta
 -\delta\,\mathcal Gw,
 \label{eq:etatarget}
\end{equation}
with the homogeneous conditions \eqref{eq:etabc} and the interconnection operator
\begin{equation}
 \mathcal G:=\PP_\varepsilon\big(\partial_x+C(\cdot)\PP_\varepsilon\big)\TT^{-1} .
 \label{eq:Gdef}
\end{equation}
The pair \eqref{eq:wtarget}, \eqref{eq:etatarget} is the closed loop written entirely in $(w,\eta)$: a feedback interconnection of the exponentially stable slow transport and the fast transport(-diffusion), in which the path from $\eta$ to $w$, through $\TT C$, is of unit order, while the return path from $w$ to $\eta$, through $\mathcal G$, and the self-loop through $\PP_\varepsilon C$, both carry the factor $\delta$. Since the loop gain is $O(\delta)$, a small-gain argument in the sense of \cite{KarafyllisKrstic2019} suggests itself; but the return path acts through both $\|w\|$ and the outflow trace $w(0,t)$ (Lemma~\ref{lem:Put} below), and no estimate of $|w(0,t)|$ by $\|w(t)\|$ exists---what the slow target supplies, through its boundary dissipation, is a bound on $\int_0^tw(0,s)^2ds$. A small-gain composition must therefore pair estimates integrated in time, and the composite Lyapunov functional below, with its optimized weight, is that composition, yielding the explicit, maximized $\delta^\star$.

\subsection{Bound on interconnection operator}
The operator $\mathcal G$ contains $\PP_\varepsilon\partial_x\TT^{-1}$, and bounding $\PP_\varepsilon\partial_x$ through the $H^1$ norm of its argument would be useless, since no bound on spatial derivatives of the state is available. The central estimate of the paper, Lemma~\ref{lem:Pux} below, shows that the derivative is absorbed entirely, leaving only the $L^2$ norm and the outflow trace of the argument. It is proved by duality: the derivative is moved onto the adjoint state, whose boundary traces are controlled by the adjoint outflow dissipation, uniformly in $\varepsilon$.

The adjoint of $\mathcal A_\varepsilon(x)$ is
\begin{equation}
 \mathcal A_\varepsilon(x)^*\psi
 =\frac{1}{\Lambda(x)}
 \Big[\varepsilon\psi''+\big(\nu(x,\cdot)\psi\big)_z+d(x,\cdot)\psi
 +\int_0^D\theta(x,\zeta,\cdot)\psi(\zeta)\,d\zeta\Big],
 \label{eq:Aepsstar}
\end{equation}
on $\{\psi\in H^2:\psi'(0)=0,\ \varepsilon\psi'(D)+\nu(x,D)\psi(D)=0\}$ for $\varepsilon>0$ and $\{\psi\in H^1:\psi(D)=0\}$ for $\varepsilon=0$, as the boundary pairing shows.

\begin{lemma}[Traces of adjoint state, uniform in $\varepsilon$]\label{lem:adjtrace}
Let $F\in L^2_x(\Hs)$ and let $\Psi(\xi):=\int_\xi^1\Phi_\varepsilon(x,\xi)^*F(x)\,dx$ be the mild solution of the backward problem $-\Psi'=\mathcal A_\varepsilon(\xi)^*\Psi+F(\xi)$ with $\Psi(1)=0$. Then $\|\Psi(\xi)\|_\Hs\le M_\Phi\|F\|$ for all $\xi$, and
\begin{equation}
 \int_0^1\Psi(\xi,0)^2\,d\xi\le N_\Psi^2\|F\|^2,
 \qquad
 \int_0^1\Psi(\xi,D)^2\,d\xi\le N_\Psi^2\|F\|^2,
 \label{eq:adjtrace}
\end{equation}
with $N_\Psi$ from \eqref{eq:NPsi}, uniformly in $\varepsilon\in[0,\bar\varepsilon]$.
\end{lemma}

\begin{proof}
The pointwise bound follows from $\|\Phi_\varepsilon^*\|=\|\Phi_\varepsilon\|\le M_\Phi$. The adjoint energy identity $-\tfrac12\frac{d}{d\xi}\|\Psi\|_\Hs^2=\langle\Psi,\mathcal A_\varepsilon^*\Psi\rangle+\langle\Psi,F\rangle$ carries the boundary computation dual to the primal one in the proof of Lemma~\ref{lem:evolution}: with the adjoint boundary conditions,
\begin{align}
 \Lambda(\xi)\langle\psi,\mathcal A_\varepsilon^*\psi\rangle_\Hs
 =&-\tfrac12\nu(\xi,0)\psi(0)^2-\tfrac12\nu(\xi,D)\psi(D)^2
 -\varepsilon\|\psi_z\|_\Hs^2
 \nonumber\\
 &+\tfrac12\langle\nu_z\psi,\psi\rangle_\Hs
 +\langle d\psi,\psi\rangle_\Hs+\langle\psi,\Theta^*\psi\rangle_\Hs ,
 \label{eq:adjdiss}
\end{align}
where $\Theta^*$ is the transposed integral operator, of the same Hilbert--Schmidt bound $\bar\theta$; at $\varepsilon=0$ the trace at $z=D$ is replaced by $\psi(D)=0$ and the identity holds with the corresponding terms absent. Integrating the identity over $(0,1)$, discarding $\tfrac12\|\Psi(0)\|^2\ge0$, and using $\Lambda\le\bar\Lambda$, $\nu\ge\un$, and the pointwise bound on $\|\Psi\|$ gives
\begin{equation}
 \frac{\un}{2\bar\Lambda}\int_0^1\big[\Psi(\xi,0)^2+\Psi(\xi,D)^2\big]d\xi
 \le\Big(\frac{\rho_1}{2}+M_R\Big)M_\Phi^2\|F\|^2+M_\Phi\|F\|^2,
 \label{eq:adjtraceproof}
\end{equation}
which is \eqref{eq:adjtrace}.
\end{proof}

\begin{lemma}[Derivative absorption by $\PP_\varepsilon$]\label{lem:Pux}
Under Assumption~\ref{as:coeff}, for every $\varepsilon\in[0,\bar\varepsilon]$ and every $h\in H^1(0,1)$,
\begin{equation}
 \|\PP_\varepsilon h'\|_{L^2_x(\Hs)}
 \le A_P\|h\|+B_P|h(0)|,
 \label{eq:Puxbound}
\end{equation}
with $A_P$, $B_P$ from \eqref{eq:AP}, \eqref{eq:BP}.
\end{lemma}

\begin{proof}
It suffices to bound $\langle\PP_\varepsilon h',F\rangle$ by $(A_P\|h\|+B_P|h(0)|)\|F\|$ for $F$ in a dense class of smooth fields compatible with the adjoint boundary conditions, for which $\Psi$ is a classical solution; the estimate extends by density. By Fubini,
\begin{equation}
 \langle\PP_\varepsilon h',F\rangle
 =\int_0^1h'(\xi)\,q(\xi)\,d\xi,
 \qquad
 q(\xi):=\langle B(\xi),\Psi(\xi)\rangle_\Hs,
 \label{eq:hdef}
\end{equation}
and, since $q(1)=0$, integration by parts in $\xi$ gives
\begin{equation}
 \langle\PP_\varepsilon h',F\rangle
 =-h(0)q(0)-\int_0^1h(\xi)q'(\xi)\,d\xi,
 \qquad
 q'=\langle\partial_\xi B,\Psi\rangle-\langle B,\mathcal A_\varepsilon^*\Psi\rangle-\langle B,F\rangle .
 \label{eq:hIBP}
\end{equation}
The middle term is computed by moving all $z$-derivatives onto $b(\xi,\cdot)=\Lambda(\xi)B(\xi)$:
\begin{align}
 \Lambda(\xi)^2\langle B,\mathcal A_\varepsilon^*\Psi\rangle
 =&\ \varepsilon\big[b\Psi_z\big]_0^D-\varepsilon\big[b_z\Psi\big]_0^D
 +\big[b\nu\Psi\big]_0^D
 \nonumber\\
 &\ +\varepsilon\langle b_{zz},\Psi\rangle
 -\langle\nu b_z,\Psi\rangle
 +\langle db,\Psi\rangle+\langle b,\Theta^*\Psi\rangle
 \nonumber\\
 =&\ -\varepsilon b_z(\xi,D)\Psi(\xi,D)
 +\big[\varepsilon b_z(\xi,0)-b(\xi,0)\nu(\xi,0)\big]\Psi(\xi,0)
 \nonumber\\
 &\ +\varepsilon\langle b_{zz},\Psi\rangle-\langle\nu b_z,\Psi\rangle
 +\langle db,\Psi\rangle+\langle b,\Theta^*\Psi\rangle,
 \label{eq:BAstar}
\end{align}
where the last equality used the adjoint boundary conditions: $\Psi_z(\xi,0)=0$ eliminates the $\varepsilon$-trace at $z=0$, and $\varepsilon b(D)\Psi_z(\xi,D)=-b(D)\nu(\xi,D)\Psi(\xi,D)$ cancels the transport trace $b(D)\nu(D)\Psi(\xi,D)$ exactly; at $\varepsilon=0$ the same expression results with the $\varepsilon$-terms absent and $\Psi(\xi,D)=0$. Now estimate the terms of \eqref{eq:hIBP}. The boundary term: $|h(0)q(0)|\le|h(0)|\,(\sqrt D\,B_0/\uL)\,M_\Phi\|F\|$, which is \eqref{eq:BP}. The volume terms: $\|\partial_\xi B\|_\Hs\le\sqrt D(B_x+B_0\Lambda_1/\uL)/\uL$, the interior part of \eqref{eq:BAstar} is bounded by $\sqrt D(\bar\varepsilon B_{zz}+\bar\nu B_z+\bar dB_0+\bar\theta B_0)\|\Psi(\xi)\|_\Hs/\uL^2$, and $|\langle B,F\rangle|\le(\sqrt D\,B_0/\uL)\|F(\xi)\|_\Hs$; using $\|\Psi(\xi)\|\le M_\Phi\|F\|$ and Cauchy--Schwarz in $\xi$, their contribution to $\int|h||q'|$ is the first two groups of \eqref{eq:AP} times $\|h\|\,\|F\|$. The trace terms of \eqref{eq:BAstar}, of sizes $(B_0\bar\nu+\bar\varepsilon B_z)|\Psi(\xi,0)|/\uL^2$ and $\bar\varepsilon B_z|\Psi(\xi,D)|/\uL^2$, are paired with $h$ by Cauchy--Schwarz in $\xi$ and bounded through Lemma~\ref{lem:adjtrace}, giving the last group of \eqref{eq:AP}.
\end{proof}

In Lemma~\ref{lem:Pux}, the derivative generates only the $L^2$ norm and the single outflow trace of the argument; applied at $h=\TT^{-1}w$ inside $\mathcal G$, that trace is $u(0)=w(0)$ by \eqref{eq:traceidentity}, which the slow transport target already dissipates. And the proof never differentiates the kernel of $\PP_\varepsilon$; this is what places the kernel jump at $\varepsilon=0$ and the parabolic smoothing at $\varepsilon>0$ under one argument.

\begin{lemma}[Bounds on $\mathcal Gw$ and $\PP_\varepsilon C\eta$]\label{lem:Put}
For $w=\TT u$ with $u\in H^1(0,1)$, and for $\eta\in L^2((0,1)\times(0,D))$,
\begin{equation}
 \|\mathcal Gw\|_{L^2_x(\Hs)}
 \le C_1\|w\|+C_2|w(0)|,
 \qquad
 \|\PP_\varepsilon C(\cdot)\eta\|_{L^2_x(\Hs)}
 \le C_3\|\eta\|,
 \label{eq:Putbound}
\end{equation}
with $C_1,C_2,C_3$ from \eqref{eq:C123}, $\|w\|$ the $L^2(0,1)$ norm, and $\|\eta\|$ the $L^2((0,1)\times(0,D))$ norm.
\end{lemma}

\begin{proof}
With $u=\TT^{-1}w$, \eqref{eq:Gdef} reads $\mathcal Gw=\PP_\varepsilon u_x+\PP_\varepsilon\,C(\cdot)\PP_\varepsilon u$. Lemma~\ref{lem:Pux}, applied at $h=u$, bounds the first term by $A_P\|u\|+B_P|u(0)|$. For the second, $|C(x)\phi|\le M_C\|\phi\|_\Hs$ gives $\|C(\cdot)\PP_\varepsilon u\|_{L^2(0,1)}\le M_CM_P\|u\|$ and hence $\|\PP_\varepsilon C\PP_\varepsilon u\|\le M_P^2M_C\|u\|$; similarly $\|\PP_\varepsilon C\eta\|\le M_PM_C\|\eta\|$. Finally $\|u\|\le M_T\|w\|$ by Lemma~\ref{lem:invert}, and $u(0)=w(0)$ by \eqref{eq:traceidentity}. Collecting constants yields \eqref{eq:C123}.
\end{proof}

\subsection{Lyapunov estimates and composite functional}
Take
\begin{equation}
 V_s=\frac12\int_0^1e^{\lambda x}\,w(x,t)^2\,dx,\qquad\lambda>0 .
 \label{eq:Vs}
\end{equation}

\begin{lemma}[Slow decay with trace dissipation]\label{lem:slow}
Along \eqref{eq:wtarget},
\begin{equation}
 \dot V_s\le-a_sV_s-\tfrac12|w(0,t)|^2+2\beta_1\sqrt{V_sV_f},
 \label{eq:Vsdot}
\end{equation}
with $a_s=\lambda$, $\beta_1$ from \eqref{eq:betas}, and $V_f$ from \eqref{eq:Vf} below.
\end{lemma}

\begin{proof}
Using \eqref{eq:wtarget}, $w(1,t)=0$, and one integration by parts,
\begin{align}
 \dot V_s
 &=\frac12\big[e^{\lambda x}w^2\big]_0^1-\frac\lambda2\int_0^1e^{\lambda x}w^2
 +\int_0^1e^{\lambda x}w\,\TT[C\eta]
 \nonumber\\
 &=-\frac12|w(0)|^2-\lambda V_s+\int_0^1e^{\lambda x}w\,\TT[C\eta] .
 \label{eq:Vsdotexact}
\end{align}
By Cauchy--Schwarz with the weight $e^{\lambda x}\le e^{\lambda}$, $\|\TT\|\le1+\bar k$, $\|C\eta\|_{L^2(0,1)}\le M_C\|\eta\|$, and $\|\eta\|^2\le e^{\sigma}\,2V_f$ from \eqref{eq:Vf},
\begin{equation}
 \left|\int_0^1e^{\lambda x}w\,\TT[C\eta]\right|
 \le\sqrt{2V_s}\;e^{\lambda/2}(1+\bar k)M_C\,\|\eta\|
 \le2(1+\bar k)M_C\,e^{(\lambda+\sigma)/2}\sqrt{V_sV_f},
 \label{eq:slowcross}
\end{equation}
which is \eqref{eq:Vsdot} with $\beta_1$ from \eqref{eq:betas}.
\end{proof}

Take
\begin{equation}
 V_f=\frac12\int_0^1\int_0^De^{-\sigma x}\,\eta(x,z,t)^2\,dz\,dx,\qquad\sigma>\sigma_0 .
 \label{eq:Vf}
\end{equation}

\begin{lemma}[Fast decay at rate $a_f/\delta$]\label{lem:fast}
Along \eqref{eq:etatarget} with \eqref{eq:etabc},
\begin{equation}
 \dot V_f
 \le-2\Big(\frac{a_f}{\delta}-\gamma\Big)V_f
 +2\beta_2\sqrt{V_sV_f}
 +\beta_0|w(0)|\sqrt{V_f},
 \label{eq:Vfdot}
\end{equation}
with $a_f,\gamma$ from \eqref{eq:asaf} and $\beta_0,\beta_2$ from \eqref{eq:betas}, uniformly in $\varepsilon\in[0,\bar\varepsilon]$.
\end{lemma}

\begin{proof}
Multiply \eqref{eq:etatarget} by $e^{-\sigma x}\eta$ and integrate over the rectangle. The two transport terms are integrated by parts:
\begin{align}
 -\frac12\int_0^D\!\int_0^1e^{-\sigma x}\Lambda\,(\eta^2)_x\,dx\,dz
 &=-\frac12\int_0^De^{-\sigma}\Lambda(1)\,\eta(1,z)^2\,dz
 \nonumber\\
 &\quad+\frac12\int_0^1\!\!\int_0^De^{-\sigma x}\big(\Lambda'-\sigma\Lambda\big)\eta^2,
 \label{eq:xIBP}\\
 -\frac12\int_0^1e^{-\sigma x}\!\int_0^D\nu\,(\eta^2)_z\,dz\,dx
 &=-\frac12\int_0^1e^{-\sigma x}\nu(x,D)\,\eta(x,D)^2\,dx
 \nonumber\\
 &\quad+\frac12\int_0^1e^{-\sigma x}\nu(x,0)\,\eta(x,0)^2\,dx
 \nonumber\\
 &\quad+\frac12\int_0^1\!\!\int_0^De^{-\sigma x}\nu_z\,\eta^2,
 \label{eq:zIBP}
\end{align}
the inflow contribution at $x=0$ vanishing by \eqref{eq:etabc}. The diffusion term contributes
\begin{equation}
 \varepsilon\int_0^1\!\!\int_0^De^{-\sigma x}\eta\,\eta_{zz}\,dz\,dx
 =\int_0^1e^{-\sigma x}\Big[\varepsilon\eta\eta_z\Big]_{z=0}^{z=D}dx
 -\varepsilon\int_0^1e^{-\sigma x}\|\eta_z(x,\cdot)\|_\Hs^2\,dx,
 \label{eq:epsfastdiss}
\end{equation}
in which the trace at $z=D$ vanishes (Neumann) and the trace at $z=0$ equals $-\nu(x,0)\eta(x,0)^2$ (zero-flux condition); at $\varepsilon=0$, \eqref{eq:epsfastdiss} is absent and $\eta(x,0)=0$. In either case the net boundary contribution at $z=0$ is $-\tfrac12\int e^{-\sigma x}\nu(x,0)\eta(x,0)^2\,dx\le0$; the contributions at $z=D$ and $x=1$ are nonpositive; and all of them, together with the dissipation $-\varepsilon\int e^{-\sigma x}\|\eta_z\|^2$, are discarded. Pointwise, $\Lambda'-\sigma\Lambda+\nu_z\le\Lambda_1+\nu_1-\sigma\uL$. The zero-order term contributes at most $\bar d\cdot2V_f$, and the nonlocal term, by Cauchy--Schwarz at each fixed $x$ with the Hilbert--Schmidt bound $\bar\theta$ (the weight depends on $x$ only, so it factors out of the $\zeta$ integration), contributes at most $\bar\theta\cdot2V_f$. Hence
\begin{equation}
 \delta\dot V_f
 \le-\big(\sigma\uL-\Lambda_1-\nu_1-2\bar d-2\bar\theta\big)V_f
 -\delta\int_0^1\!\!\int_0^De^{-\sigma x}\eta\,\big(\PP_\varepsilon C\eta+\mathcal Gw\big)\,dz\,dx
 =-2a_fV_f-\delta\,\Xi,
 \label{eq:Vfdotexact}
\end{equation}
with $a_f$ from \eqref{eq:asaf}. For the coupling term, Cauchy--Schwarz with $e^{-\sigma x}\le1$ and Lemma~\ref{lem:Put} give
\begin{equation}
 |\Xi|
 \le\sqrt{2V_f}\,\big(\|\mathcal Gw\|+\|\PP_\varepsilon C\eta\|\big)
 \le\sqrt{2V_f}\Big(C_1\sqrt{2V_s}+C_2|w(0)|+C_3e^{\sigma/2}\sqrt{2V_f}\Big),
 \label{eq:Xibound}
\end{equation}
where $\|w\|\le\sqrt{2V_s}$ (the weight $e^{\lambda x}\ge1$) and $\|\eta\|\le e^{\sigma/2}\sqrt{2V_f}$ were used. Dividing \eqref{eq:Vfdotexact} by $\delta$ and inserting \eqref{eq:Xibound} yields \eqref{eq:Vfdot} with $\gamma=C_3e^{\sigma/2}$, $\beta_2=C_1$, $\beta_0=\sqrt2\,C_2$.
\end{proof}

Take
\begin{equation}
 V=V_s+\kappa V_f,\qquad\kappa>0 .
 \label{eq:Vcomp}
\end{equation}
Adding \eqref{eq:Vsdot} and $\kappa\times$\eqref{eq:Vfdot},
\begin{equation}
 \dot V
 \le-a_sV_s-\tfrac12|w(0)|^2
 -2\kappa\Big(\frac{a_f}{\delta}-\gamma\Big)V_f
 +2(\beta_1+\kappa\beta_2)\sqrt{V_sV_f}
 +\kappa\beta_0|w(0)|\sqrt{V_f}.
 \label{eq:Vdotpre}
\end{equation}
Young's inequality eliminates the trace:
\begin{equation}
 -\tfrac12|w(0)|^2+\kappa\beta_0|w(0)|\sqrt{V_f}
 \le\frac{\kappa^2\beta_0^2}{2}\,V_f .
 \label{eq:tracecomplete}
\end{equation}
The right side of \eqref{eq:Vdotpre} is then a quadratic form in $\big(\sqrt{V_s},\sqrt{V_f}\big)$, negative definite if and only if
\begin{equation}
 \frac{a_f}{\delta}-\gamma
 >\frac{(\beta_1+\kappa\beta_2)^2}{2\kappa a_s}
 +\frac{\kappa\beta_0^2}{4}
 =\frac{\beta_1^2}{2a_s}\,\frac1\kappa
 +\frac{\beta_1\beta_2}{a_s}
 +\left(\frac{\beta_2^2}{2a_s}+\frac{\beta_0^2}{4}\right)\kappa .
 \label{eq:kappacond}
\end{equation}
The right side is of the form $A/\kappa+B\kappa+\beta_1\beta_2/a_s$ with $A=\beta_1^2/(2a_s)$ and $B=\big(\beta_2^2+a_s\beta_0^2/2\big)/(2a_s)$, minimized at $\kappa=\sqrt{A/B}$ with minimum $2\sqrt{AB}+\beta_1\beta_2/a_s$. With
\begin{equation}
 \widehat\beta_2:=\sqrt{\beta_2^2+\frac{a_s\beta_0^2}{2}},
 \label{eq:betahat}
\end{equation}
the optimal weight and the minimized threshold are
\begin{equation}
 \kappa^\star=\frac{\beta_1}{\widehat\beta_2},
 \qquad
 \min_{\kappa>0}\left[\frac{(\beta_1+\kappa\beta_2)^2}{2\kappa a_s}
 +\frac{\kappa\beta_0^2}{4}\right]
 =\frac{\beta_1}{a_s}\big(\beta_2+\widehat\beta_2\big),
 \label{eq:kappastar}
\end{equation}
so that \eqref{eq:kappacond} at $\kappa=\kappa^\star$ is exactly $\delta<\delta^\star(\lambda,\sigma)$ with $\delta^\star$ from \eqref{eq:deltastar}.

Fix $\lambda,\sigma$ and $\delta<\delta^\star(\lambda,\sigma)$, and set $\kappa=\kappa^\star$, $b_\star:=\beta_1+\kappa^\star\beta_2$, and
\begin{equation}
 m:=2\kappa^\star\Big(\frac{a_f}{\delta}-\gamma\Big)-\frac{(\kappa^\star)^2\beta_0^2}{2}>0,
 \qquad
 a_sm>b_\star^2,
 \label{eq:margin}
\end{equation}
the last inequality being \eqref{eq:kappacond}. The quadratic form $-a_sX^2+2b_\star XY-mY^2$ is then bounded above by $-\mu(X^2+Y^2)$ with
\begin{equation}
 \mu=\frac12\Big(a_s+m-\sqrt{(a_s-m)^2+4b_\star^2}\Big)>0,
 \label{eq:mu}
\end{equation}
so \eqref{eq:Vdotpre}, \eqref{eq:tracecomplete} give
\begin{equation}
 \dot V\le-\mu\,(V_s+V_f)\le-\alpha_0V,
 \qquad
 \alpha_0:=\mu\min\{1,1/\kappa^\star\},
 \label{eq:Vdecay}
\end{equation}
hence $V(t)\le e^{-\alpha_0t}V(0)$.

Finally, $V$ is equivalent to the physical energy. From $w=\TT u$, $u=\TT^{-1}w$, $\eta=v-\PP_\varepsilon u$, and Lemmas \ref{lem:evolution} and \ref{lem:invert},
\begin{align}
 \|w\|&\le(1+\bar k)\|u\|,\qquad
 \|u\|\le M_T\|w\|,
 \nonumber\\
 \|\eta\|&\le\|v\|+M_P\|u\|,\qquad
 \|v\|\le\|\eta\|+M_PM_T\|w\|,
 \label{eq:equiv1}
\end{align}
while $1\le e^{\lambda x}\le e^{\lambda}$ and $e^{-\sigma}\le e^{-\sigma x}\le1$ bound the weights. Consequently there exist $m_1,m_2>0$, explicit in $\bar k,M_T,M_P,\lambda,\sigma,\kappa^\star$, with
\begin{equation}
 m_1\big(\|u\|^2+\|v\|^2\big)\le V\le m_2\big(\|u\|^2+\|v\|^2\big),
 \label{eq:equiv2}
\end{equation}
and \eqref{eq:Vdecay} gives \eqref{eq:mainestimate} with $M=m_2/m_1$ and $\alpha=\alpha_0$. Since every constant depends on $\varepsilon$ only through $\bar\varepsilon$, and the controller does not depend on $(\lambda,\sigma)$, the conclusion holds for every $\varepsilon\in[0,\bar\varepsilon]$ and every $\delta<\Delta$. \hfill$\square$

\section{Example: Constant-Coefficient Counterconvection}\label{sec:example}
Take $\Lambda\equiv1$, $\nu\equiv1$, $d\equiv0$, $\theta\equiv0$, $b=b(z)$, $c=c(z)$, $D>1$, and $\bar\varepsilon=0$, the setting of Corollary~\ref{cor:transport}. Then $\rho\equiv1$, characteristics are the lines $z-x=\mathrm{const}$, and
\begin{equation}
 (\PP_0 u)(x,z)=\int_{\max\{0,\,x-z\}}^{x}b(z-x+\xi)\,u(\xi)\,d\xi,
 \label{eq:Pexplicit}
\end{equation}
so the reduced kernel is the convolution
\begin{equation}
 f_0(x,\xi)=f(x-\xi),
 \qquad
 f(r)=\int_r^Dc(z)\,b(z-r)\,dz ,
 \label{eq:kconv}
\end{equation}
where $D>1$ keeps the corner characteristic inside the rectangle, so that $f\in C^1([0,1])$. For $D<1$, the kernel \eqref{eq:kconv} satisfies $f(r)=0$ for $r>D$, so in the reduced model $u_t(x,t)=u_x(x,t)+\int_{x-D}^xf(x-\xi)u(\xi,t)\,d\xi$: the state $u(\xi,t)$ with $\xi<x-D$ does not enter. No fast subsystem that is a spatial ODE produces such an $f$, and neither does an ensemble ($\nu=0$, $\varepsilon=0$): their kernels are sums or integrals of exponentials in $r$, analytic, and an analytic function vanishing for $r>D$ vanishes for all $r$. Volterra kernels vanishing for $r>D$ in the class of \cite{KrsticSmyshlyaev2008} thus arise from transport in $z$. The general constants specialize to
\begin{equation}
 \rho_1=0,\quad M_R=0,\quad M_\Phi=1,\quad
 M_P=\sqrt D\,B_0,\quad
 N_\Psi=\sqrt2,\quad
 \sigma_0=0,
 \label{eq:exconst1}
\end{equation}
\begin{equation}
 A_P=\sqrt D\,(B_0+B_z)+\sqrt2\,B_0,
 \qquad
 B_P=\sqrt D\,B_0 .
 \label{eq:exconst2}
\end{equation}
The mechanism of Lemma~\ref{lem:Pux} is visible in the elementary computation: for $0<z\le x$, integration by parts in \eqref{eq:Pexplicit} gives
\begin{equation}
 (\PP_0 u_x)(x,z)=b(z)u(x)-b(0)u(x-z)-\int_{x-z}^xb'(z-x+\xi)u(\xi)\,d\xi,
 \label{eq:excase1}
\end{equation}
the middle term being the jump of the kernel across the characteristic $z=x-\xi$, while for $x<z<D$
\begin{equation}
 (\PP_0 u_x)(x,z)=b(z)u(x)-b(z-x)u(0)-\int_0^xb'(z-x+\xi)u(\xi)\,d\xi,
 \label{eq:excase2}
\end{equation}
the middle term being the trace contribution that becomes $B_P|u(0)|$. A direct integration by parts on the kernel, available at $\varepsilon=0$ where the jump is explicit, yields from these formulas the slightly sharper $A_P=(\sqrt D+1)B_0+\sqrt DB_z$; the proof of Theorem~\ref{thm:main} carries the additional factor $\sqrt2$ on the jump contribution because it never differentiates the kernel. The physical picture is a bulk species $u$ convecting toward decreasing $x$, and a fast secondary species $v$ counterconvecting toward increasing $x$ while progressing through the internal coordinate $z$, injected from $u$ with intensity $b(z)$ and fed back through the aggregate $\int_0^Dc(z)v\,dz$. The apparent self-feeding of $u$ from ahead in the reduced model \eqref{eq:genreslow} is mediated by this rapidly counterconvecting species, whose unresolved dynamics appear to $u$, in the limit $\delta\to0$, exactly as the spatial Volterra operator with kernel \eqref{eq:kconv}. Theorem~\ref{thm:main} certifies the range $\delta<\Delta$ over which the backstepping feedback designed for \eqref{eq:kconv} stabilizes the full two-dimensional plant.

\section{Concluding Remarks}
The architecture parallels the singular-perturbation analysis of the thermal convection loop \cite{VazquezKrstic2006}: introduce the error between the fast state and its quasi-steady value, derive a stable boundary-layer subsystem, combine its Lyapunov functional with that of the reduced slow target, and use smallness of the singular perturbation parameter to dominate the cross couplings. The present result differs in three respects. The fast subsystem is a two-dimensional PDE on a rectangle rather than a one-dimensional fast PDE; the relative weight of the two energies is chosen analytically, in \eqref{eq:kappastar}, to maximize the admissible $\delta$ within the Lyapunov estimates rather than being fixed a priori; and the pure-transport and transport-diffusion fast subsystems fall under one theorem with one proof, because the central estimate never differentiates the kernel of the quasi-steady solution operator.

In the ensembles of \cite{AlleaumeKrstic2025}, there is no transport or diffusion in the ensemble variable, an exact design exists, and the difficulty is concentrated in the kernel equations of the full coupled system. Here, no scalar input can be expected to control the two-dimensional fast state exactly, and no design of that kind exists; the feedback is the scalar design of \cite{KrsticSmyshlyaev2008} on the reduced model, and the difficulty is concentrated in the analysis, in the motion of the quasi-steady state and in the absorption of the spatial derivative in Lemma~\ref{lem:Pux}. The time-scale separation is what makes stabilization available at all: what no scalar input achieves exactly, wide separation achieves exponentially, for every $\delta$ below the bound of Theorem~\ref{thm:main}.

An open extension is the boundary coupling $v(0,z,t)=b_0(z)u(0,t)$ of \cite[eq.~(4)]{AlleaumeKrstic2025}. The reduced model then acquires a $u(0,t)$ term of the class of \cite{KrsticSmyshlyaev2008} and the design carries over; but the motion of the quasi-steady state places the first-order trace $u_t(0,t)$ in the fast error equation, which the $L^2$ analysis does not control, so the extension requires an $H^1$ analysis or smallness of $b_0$.

\appendix
\section*{Appendix}
\addcontentsline{toc}{section}{Appendix}
The following lemmas certify the assertions of Sections~\ref{sec:design} and \ref{sec:main} and the well-posedness claim of Theorem~\ref{thm:main}.

We first record the characteristics of the transport part, used in the $\varepsilon=0$ construction. For $0\le\xi\le x\le1$ and $z\in[0,D]$, let $s\mapsto Z(s;x,z)$ denote the characteristic of $\rho$ through $(x,z)$,
\begin{equation}
 \frac{dZ}{ds}=\rho(s,Z),\qquad Z(x;x,z)=z,
 \label{eq:charODE}
\end{equation}
solved backward from $s=x$ until either $s=\xi$ is reached or the characteristic exits through $z=0$; in the latter case $\xi^*(x,z)$ denotes the exit abscissa, $Z(\xi^*;x,z)=0$. Since $\rho\ge\rho_0:=\un/\bar\Lambda>0$, characteristics are strictly increasing in $s$ and the exit abscissa, when it exists, is unique. Differentiating \eqref{eq:charODE} with respect to the datum $z$ gives the variational equation $\tfrac{d}{ds}\partial_zZ=\rho_z(s,Z)\partial_zZ$ with $\partial_zZ(x;x,z)=1$, hence
\begin{equation}
 \partial_z Z(s;x,z)
 =\exp\!\Big(-\!\int_s^x\rho_z(\tau,Z(\tau;x,z))\,d\tau\Big)
 \in\big[e^{-\rho_1(x-s)},\,e^{\rho_1(x-s)}\big].
 \label{eq:charJac}
\end{equation}

\begin{lemma}[Spatial evolution family; continuity of $f_\varepsilon$]\label{lem:evolution}
Under Assumption~\ref{as:coeff}, for every $\varepsilon\in[0,\bar\varepsilon]$ the homogeneous problem $V_x=\mathcal A_\varepsilon(x)V$, $V(\xi)=\phi$, is well posed on $\xi\le x\le1$ for every $\phi\in\Hs$, and its solution operators $\Phi_\varepsilon(x,\xi)\in\mathcal L(\Hs)$ form an evolution family, with $\Phi_\varepsilon(x,x)=I$, $\Phi_\varepsilon(x,s)\Phi_\varepsilon(s,\xi)=\Phi_\varepsilon(x,\xi)$, and $\|\Phi_\varepsilon(x,\xi)\|_{\mathcal L(\Hs)}\le M_\Phi$, uniformly in $\varepsilon$. Moreover $\PP_\varepsilon$ from \eqref{eq:Pgen} is the unique mild solution map of \eqref{eq:Hspatial}, $\|\PP_\varepsilon\|\le M_P$ as an operator from $L^2(0,1)$ to $L^2_x(\Hs)$, uniformly in $\varepsilon$; $\PP_\varepsilon u$ satisfies the homogeneous conditions \eqref{eq:gen-zbc}, and $(\PP_\varepsilon u)(0,z)=0$; and $f_\varepsilon$ from \eqref{eq:Kgen} is continuous on $T$.
\end{lemma}

\begin{proof}
The uniform bound is a dissipativity computation. With $\phi$ in the domain of $\mathcal A_\varepsilon(x)$, for $\varepsilon>0$,
\begin{equation}
 \varepsilon\langle\phi,\phi''\rangle_\Hs
 =\varepsilon\phi(D)\phi'(D)-\varepsilon\phi(0)\phi'(0)-\varepsilon\|\phi'\|_\Hs^2
 =-\nu(x,0)\phi(0)^2-\varepsilon\|\phi'\|_\Hs^2,
 \label{eq:eps-diss1}
\end{equation}
\begin{equation}
 -\langle\phi,\nu\phi'\rangle_\Hs
 =-\tfrac12\nu(x,D)\phi(D)^2+\tfrac12\nu(x,0)\phi(0)^2
 +\tfrac12\langle\nu_z\phi,\phi\rangle_\Hs,
 \label{eq:eps-diss2}
\end{equation}
so the boundary contributions at $z=0$ sum to $-\tfrac12\nu(x,0)\phi(0)^2\le0$ and those at $z=D$ to $-\tfrac12\nu(x,D)\phi(D)^2\le0$; for $\varepsilon=0$, \eqref{eq:eps-diss1} is absent, $\phi(0)=0$, and \eqref{eq:eps-diss2} alone gives the same conclusion. Since $\nu_z/\Lambda=\rho_z$,
\begin{equation}
 \langle\phi,\mathcal A_\varepsilon(x)\phi\rangle_\Hs
 \le\Big(\frac{\rho_1}{2}+M_R\Big)\|\phi\|_\Hs^2,
 \label{eq:eps-diss3}
\end{equation}
and Gronwall applied to $\frac{d}{dx}\|V\|_\Hs\le(\rho_1/2+M_R)\|V\|_\Hs+\|B(x)\|_\Hs|u(x)|$ gives \eqref{eq:Phibound} and the variation-of-constants formula \eqref{eq:Pgen} on classical solutions, extended by density.

Existence of the family: for each fixed $\varepsilon>0$ it follows from Lions' theory of nonautonomous forms. The domain of $\mathcal A_\varepsilon(x)$ varies with $x$ through the Robin coefficient $\nu(x,0)$, but the form domain does not: the associated form on $H^1(0,D)\times H^1(0,D)$,
\begin{equation}
 a_\varepsilon(x;\phi,\psi)
 =\frac{1}{\Lambda(x)}
 \Big[\varepsilon\langle\phi',\psi'\rangle_\Hs
 +\nu(x,0)\phi(0)\psi(0)
 +\langle\nu\phi',\psi\rangle_\Hs
 -\langle d\phi,\psi\rangle_\Hs
 -\langle\Theta\phi,\psi\rangle_\Hs\Big],
 \label{eq:form}
\end{equation}
is bounded and satisfies a G\aa rding inequality on the fixed space $H^1(0,D)$, uniformly in $x$, with continuous dependence on $x$ (measurable and bounded suffices for Lions' construction), so the family exists with solutions in $C([\xi,1];\Hs)$, which gives strong continuity; the $\varepsilon$-uniform bound \eqref{eq:Phibound} then comes from \eqref{eq:eps-diss1}--\eqref{eq:eps-diss3} and not from the generation theorem, whose constants may degenerate as $\varepsilon\downarrow0$. For $\varepsilon=0$ the zero-inflow transport family is explicit,
\begin{equation}
 \big(\Phi_{0,0}(x,\xi)\phi\big)(z)
 =\begin{cases}
 \phi\big(Z(\xi;x,z)\big), & \text{if the characteristic reaches }s=\xi\text{ in }z>0,\\
 0, & \text{if it exits through }z=0\text{ at some }\xi^*\in(\xi,x],
 \end{cases}
 \label{eq:Phi0def}
\end{equation}
with the evolution property from uniqueness of characteristics; substituting $\zeta=Z(\xi;x,z)$ on the survival set and using \eqref{eq:charJac},
\begin{equation}
 \|\Phi_{0,0}(x,\xi)\phi\|_\Hs^2
 =\int\phi(\zeta)^2\,\frac{dz}{d\zeta}\,d\zeta
 \le e^{\rho_1(x-\xi)}\|\phi\|_\Hs^2,
 \label{eq:Phi0bound}
\end{equation}
and the full family at $\varepsilon=0$ is the unique solution of the Volterra operator equation
\begin{equation}
 \Phi_0(x,\xi)=\Phi_{0,0}(x,\xi)+\int_\xi^x\Phi_{0,0}(x,s)\mathcal R(s)\Phi_0(s,\xi)\,ds,
 \label{eq:PhiDuhamel}
\end{equation}
obtained by successive approximation, uniformly convergent in operator norm on $0\le\xi\le x\le1$.

Finally, since $\|B(\xi)\|_\Hs\le\sqrt D\,B_0/\uL$, the kernel $G_\varepsilon(x,\xi):=\Phi_\varepsilon(x,\xi)B(\xi)$ obeys
\begin{equation}
 \|G_\varepsilon(x,\xi)\|_\Hs\le M_\Phi\frac{\sqrt D\,B_0}{\uL}=M_P,
 \label{eq:MG}
\end{equation}
so $\|(\PP_\varepsilon u)(x)\|_\Hs\le M_P\int_0^x|u|\le M_P\|u\|$ and $\|\PP_\varepsilon u\|_{L^2_x(\Hs)}\le M_P\|u\|$. The boundary and initial conditions hold by construction of the mild solution.

Finally, continuity of $f_\varepsilon$: writing
\begin{equation}
 \Phi_\varepsilon(x,\xi)B(\xi)-\Phi_\varepsilon(x',\xi')B(\xi')
 =\Phi_\varepsilon(x,\xi)\big[B(\xi)-B(\xi')\big]
 +\big[\Phi_\varepsilon(x,\xi)-\Phi_\varepsilon(x',\xi')\big]B(\xi'),
 \label{eq:fcont}
\end{equation}
the first term is $O(\|B(\xi)-B(\xi')\|_\Hs)$ by \eqref{eq:Phibound}, with $\xi\mapsto B(\xi)$ continuous into $\Hs$ since $b\in C^1$, and the second tends to zero because the family, applied to the fixed vector $B(\xi')$, is jointly strongly continuous---for $\varepsilon=0$ this is $L^2$-continuity of composition with the characteristic flow \eqref{eq:Phi0def}, unaffected by the moving cut-off, and for $\varepsilon>0$ it is the $C([\xi,1];\Hs)$ regularity above. Pairing \eqref{eq:fcont} with $c(x,\cdot)$, continuous into $\Hs$, gives continuity of $f_\varepsilon$ on $T$; the jump of the $\Hs$-valued kernel in $z$ is integrated out by $C(x)$.
\end{proof}

\begin{lemma}[Solvability of kernel equation]\label{lem:kernel}
For every $\varepsilon\in[0,\bar\varepsilon]$, the kernel equation \eqref{eq:genq} has a unique solution $k_\varepsilon\in C(T)$, continuously differentiable along the characteristic direction $(1,1)$, with $|k_\varepsilon(x,\xi)|\le\bar k$ on $T$.
\end{lemma}

\begin{proof}
Integrating \eqref{eq:genq} along the characteristic $x-\xi=\mathrm{const}$ from the boundary $\xi=0$, where $k_\varepsilon=0$, gives the equivalent integral equation
\begin{equation}
 k_\varepsilon(x,\xi)
 =-\int_0^\xi f_\varepsilon(x-\xi+s,s)\,ds
 +\int_0^\xi\!\int_s^{x-\xi+s}k_\varepsilon(x-\xi+s,\sigma)f_\varepsilon(\sigma,s)\,d\sigma\,ds .
 \label{eq:qint}
\end{equation}
Define $k^{(0)}:=0$ and let $k^{(n+1)}$ be the right side of \eqref{eq:qint} evaluated at $k^{(n)}$, with $\Delta^{(n)}:=k^{(n+1)}-k^{(n)}$. Then $|\Delta^{(0)}(x,\xi)|\le\bar f\xi\le\bar fx$. Assume $|\Delta^{(n-1)}(x',\sigma)|\le\bar f^{\,n}(x')^{n}/n!$ for all $(x',\sigma)\in T$. Since the inner integration variable $\sigma$ in \eqref{eq:qint} runs over an interval of length $x-\xi\le1$ and the first argument of $k_\varepsilon$ there equals $x-\xi+s$,
\begin{equation}
 |\Delta^{(n)}(x,\xi)|
 \le\bar f\int_0^\xi\frac{\bar f^{\,n}(x-\xi+s)^{n}}{n!}\,ds
 \le\frac{\bar f^{\,n+1}}{n!}\cdot\frac{x^{\,n+1}}{n+1}
 =\frac{\bar f^{\,n+1}x^{\,n+1}}{(n+1)!}\,.
 \label{eq:qinduction}
\end{equation}
The series $\sum_n\Delta^{(n)}$ therefore converges uniformly on $T$ to a continuous solution of \eqref{eq:qint} bounded by $e^{\bar fx}-1\le\bar k$, and the same estimate applied to the difference of two solutions gives uniqueness. Differentiability along $(1,1)$ follows from \eqref{eq:qint}, whose integrand is continuous.
\end{proof}

\begin{lemma}[Inverse of backstepping transformation]\label{lem:invert}
$\TT=I-Q$, with $Q$ the Volterra integral operator of kernel $k_\varepsilon$, is boundedly invertible on $L^2(0,1)$, $\TT^{-1}=I+P$, where $P$ is a Volterra integral operator with continuous kernel $p$ satisfying $|p|\le\bar k\,e^{\bar k}$; consequently $\|\TT\|\le1+\bar k$ and $\|\TT^{-1}\|\le M_T$, uniformly in $\varepsilon$.
\end{lemma}

\begin{proof}
The iterated kernels of $Q$ satisfy $|k_n(x,\xi)|\le\bar k^{\,n}(x-\xi)^{n-1}/(n-1)!$ by induction on the composition integral, so the Neumann series $\sum_{n\ge1}Q^n$ has kernel $p=\sum_nk_n$ converging uniformly with $|p|\le\bar k\,e^{\bar k}$. A Volterra operator with kernel bounded by $\bar p$ has $L^2(0,1)$ norm at most $\bar p$ by the Cauchy--Schwarz estimate used for $\PP_\varepsilon$ in Lemma~\ref{lem:evolution}, which gives the bounds.
\end{proof}

\begin{lemma}[Closed-loop well-posedness]\label{lem:wellposed}
Under Assumption~\ref{as:coeff}, for every $\varepsilon\in[0,\bar\varepsilon]$, $\delta>0$, and all initial data $(u_0,v_0)\in L^2(0,1)\times L^2((0,1)\times(0,D))$, the closed loop \eqref{eq:gen-u}--\eqref{eq:gen-zbc}, \eqref{eq:genU} has a unique mild solution, continuous in time into $L^2(0,1)\times L^2((0,1)\times(0,D))$.
\end{lemma}

\begin{proof}
Integrating the $u$-equation along its unit-speed characteristics,
\begin{eqnarray}
 u(x,t)&=&u_0(x+t)\,\mathbf 1\{x+t\le1\}
 +U\big(t-(1-x)\big)\,\mathbf 1\{x+t>1\}
 \nonumber\\
 &&+\int_{\max\{0,\,t-(1-x)\}}^{t}\big(C(\cdot)v\big)(x+t-s,s)\,ds,
 \label{eq:uchar}
\end{eqnarray}
with $U(t)=\int_0^1k_\varepsilon(1,\xi)u(\xi,t)\,d\xi$. On $C([0,T];L^2(0,1))$, consider the map taking $u$ to the right side of \eqref{eq:uchar}, in which $v$ is the mild solution of \eqref{eq:gen-v} driven by $b\,u$, obtained by composing the $x$-characteristic flow of speed $\Lambda/\delta$ with the dynamics in $z$ (characteristics for $\varepsilon=0$, the form solutions above for $\varepsilon>0$) and a Duhamel iteration for the bounded remaining terms; this solution map is affine in $u$, with Lipschitz constant $O(T)$ into $C([0,T];L^2)$. The feedback term of \eqref{eq:uchar} is delayed, $U$ being evaluated at times $\le t-(1-x)$, and contributes, by Cauchy--Schwarz in $x$, at most $(\int_0^t|U|^2)^{1/2}\le\sqrt t\,\bar k\sup_{[0,t]}\|u\|$. The map is therefore a contraction for $T$ small, with $T$ depending only on $\bar k$, $M_C$, $\delta$, and the plant bounds; its fixed point, with the associated $v$, is the unique mild solution on $[0,T]$, and concatenation extends it to all $t\ge0$.
\end{proof}

\medskip\paragraph{Acknowledgment.}
The author's problems, ideas, and results were developed with the assistance of Claude and ChatGPT in final theorem formulation, proofs, and drafting throughout the paper, under the author's correction and complete verification.

This work was funded by AFOSR grant FA9550-23-1-0535 and NSF grant ECCS-2151525.

\end{document}